\documentclass[pdftex,journal,twoside,web]{ieeecolor} %journal version
\usepackage{generic}
\usepackage{cite}
\usepackage{amsmath,amsfonts}
\usepackage[nolist]{acronym}

\usepackage{amsmath}    % For mathematical expressions
\usepackage{pgfplots}   % For creating plots
\pgfplotsset{compat=1.16} % Use the version that matches your installation
\usepackage{xcolor}     % For color support
\usepackage[nolist]{acronym}
\usepackage{xcolor}
\usepackage{amsmath}
\usepackage{amssymb}
\usepackage{mathtools}
\usepackage{subcaption}
\usepackage{hyperref}
\hypersetup{colorlinks=true, linkcolor=black}
\usepackage[ruled]{algorithm2e}
\usepackage{mathrsfs}
\newcommand{\R}{\mathbb{R}}

\DeclarePairedDelimiter{\norm}{\lVert}{\rVert}

\DeclarePairedDelimiterX{\inp}[2]{\langle}{\rangle}{#1, #2}

\newtheorem{thm}{Theorem}[section]

\newtheorem{defn}{Definition}[section]

\newtheorem{rmk}{Remark} 

\newcommand{\ab}{\alpha \beta}

\usepackage{xcolor,soul,framed} %,caption

\colorlet{shadecolor}{yellow}
\usepackage{graphicx}
\DeclareGraphicsExtensions{.pdf,.jpeg,.png}

\usepackage{bbm}
\newcommand{\cm}{\checkmark}
\DeclareMathAlphabet{\mathbbb}{U}{bbold}{m}{n}
\usepackage{array}
\usepackage{aligned-overset}
\usepackage{flushend}
\usepackage{physics}
\usepackage{cite}
\usepackage{makecell}
\usepackage{threeparttable}
\usepackage{comment}
\usepackage{colortbl}
\let\labelindent\relax
\usepackage{enumitem}

\definecolor{color_hybrid_angle_control}{RGB}{247,144,71}
\def\BibTeX{{\rm B\kern-.05em{\sc i\kern-.025em b}\kern-.08em
    T\kern-.1667em\lower.7ex\hbox{E}\kern-.125emX}}
\makeatletter
\renewcommand*{\@textcolor}[3]{%
  \protect\leavevmode
  \begingroup
    \color#1{#2}#3%
  \endgroup
}
\makeatother

\begin{document}

\title{Network-Independent Incremental Passivity Conditions for Grid-Forming Inverter Control}

\author{Jared Miller, \IEEEmembership{Member, IEEE}, Maitraya Avadhut Desai, \IEEEmembership{Student Member, IEEE}, Xiuqiang He, \IEEEmembership{Member, IEEE}, Roy S. Smith \IEEEmembership{Life Fellow, IEEE} and Gabriela Hug \IEEEmembership{Senior Member, IEEE}
% \thanks{ This paper was submitted for review on June 17, 2025.}
\thanks{
J. Miller is with the Chair of Mathematical Systems Theory, Department of Mathematics,  University of Stuttgart, Stuttgart, Germany (e-mail: jared.miller@imng.uni-stuttgart.de).}
\thanks{J. Miller and R. S. Smith are with the Automatic Control Laboratory (IfA),  ETH Z\"{u}rich, Physikstrasse 3, 8092, Z\"{u}rich, Switzerland (e-mail: \{jarmiller, rsmith\}@control.ee.ethz.ch).}
\thanks{M. A. Desai and G. Hug are with the Power Systems Laboratory (PSL),  ETH Z\"{u}rich, Physikstrasse 3, 8092, Z\"{u}rich, Switzerland (e-mail: \{desai, hug\}@eeh.ee.ethz.ch.}
\thanks{X. He is with the Automation of Department,  Tsinghua University, 100084, Beijing, China (e-mail: hxq19@tsinghua.org.cn.}}
\maketitle

\begin{abstract}
\label{sec:abstract}
Grid-forming inverters control the power transfer between the AC and DC sides of an electrical grid while maintaining the frequency and voltage of the AC side. This paper focuses on ensuring large-signal stability of an electrical grid with inverter-interfaced renewable sources. 
We prove that the Hybrid-Angle Control (HAC) scheme for grid-forming inverters can exhibit incremental passivity properties between current and voltage at both the AC and DC ports. This incremental passivity can be certified through decentralized conditions.
Inverters operating under HAC can, therefore, be connected to other passive elements (e.g. transmission lines) with an immediate guarantee of global transient stability regardless of the network topology or parameters. Passivity of Hybrid Angle Control is also preserved under small-signal (linearized) analyses, in contrast to conventional proportional droop laws that are passivity-short at low frequencies. Passivity and interconnected-stability properties are demonstrated through an example case study.

\end{abstract}
\begin{IEEEkeywords}
% \begin{keywords}
Passivity, Inverter, Grid-Forming, Multi-Inverter, Electronic, Nonlinear Control
% \end{keywords}
\end{IEEEkeywords}

\section{Introduction}
\label{sec:introduction}

The energy transition requires the replacement of carbon-intensive generators with low-carbon renewable energy technologies. The previous  power grid involves a centralized top-down hierarchy, in which energy is generated by large-scale synchronous generators (e.g. hydropower, coal-fired power plants, methane gas turbines, nuclear generation). These traditional large-scale sources output AC power at a synchronous frequency of 50 or 60 Hz. In contrast, many forms of renewable energy produce either DC power or AC power at non-synchronous frequencies, thus necessitating the presence of power inverters.
It is therefore vital to have certificates of stability and performance for a power grid involving both the AC and DC sides, and to ensure that stability will be ensured in a decentralized manner. Grid codes and interconnection rules serve as the primary mechanism for maintaining stability and performance \cite{photovoltaics2018ieee}.

The inverters interlinking the DC and AC sides of the power grid are power electronic devices that produce AC signals based on Pulse-Width-Modulation (PWM) inputs, for which the switching frequency is kept much higher than the nominal grid frequency \cite{hart2011power}. Proper inverter control can ensure the safety and stability of an AC/DC power network. However, designing inverter control strategies is one of the central challenges in delivering a reliable distributed energy system \cite{denis2015review, kroposki2017achieving}.
Inverter control is typically divided into grid-following or grid-forming control \cite{dorfler2023control}. Grid-following control accepts frequency and voltage setpoint inputs (where the frequency is measured through phase-locked loops) and outputs a voltage waveform with specified active and reactive power.
In contrast, grid-forming control accepts setpoint power inputs, and outputs a voltage waveform at a frequency and magnitude to meet the given power requirements \cite{lasseter2019grid}. Grid-forming control is capable of creating a stable voltage signal under islanded operation and recovering after faults, whereas grid-following control must always remain connected to a stable grid frequency. Instances of grid-forming control include droop and virtual synchronous machine control \cite{DarcoVSM}, virtual oscillator control \cite{johnson2013synchronization, colombino2017global},  matching control \cite{arghir2018grid}, and \ac{HAC} \cite{tayyebi2023hac}.

An example of a desirable specification is passivity \cite{schiff1971passivity, van2000l2}. Because the negative feedback interconnection between two passive components remains passive (by the Passivity Theorem), arbitrary interconnection of passive elements in a power grid will remain passive. This passivity property is widely used in DC microgrid control, in which multiple prosumers operating on a common network can possess a certificate of stability \cite{Cucuzzella2019passivity,NAHATA2020108770}. Passivity \cite{nercwhitepaper,fingrid2023} (and, more generally, positive damping \cite{kroposki2022unifi}) is suggested as a desirable property for inverter-based generation.

Unfortunately, global passivity results generally do not apply when considering the AC side of the power grid. The work in \cite{chen2024limitations} reviews obstacles towards the use of passivity as a stabilization method in AC grids, specifically noting that grid-forming inverters are generally non-passive at low frequencies \cite{beza2019identification, zhao2022low}.
Table \ref{tab:comparison-passivity} compares the passivity properties of the existing AC / DC grid-forming inverter control. In particular, dVOC \cite{colombino2017global, he2023nonlinear} may be passivity-short, and thus requires a sufficiently passive admittance matrix for the interconnection in order to ensure stability of the closed-loop network.

\begin{table*}
\centering
\caption{Passivity Properties of Grid-Forming Control}
\arrayrulecolor{black!10}
\begin{tabular}{lcccccc}
\arrayrulecolor{black}
\hline\hline
Scheme & Input       & Output                                        & Nonlinear   & Global & Islanded   & Network-Independent                 \\ \hline
\textbf{HAC}  \cite{tayyebi2023hac}   & $i_{dc}, i_{\ab}$         & $v_{dc}, v_{\ab}$                                           & \checkmark                          &   \checkmark        &                   \checkmark     &  \checkmark    \\
Matching (local $dq$) \cite{arghir2018energy}& $i_{dc}, i_{dq}$         & $v_{dc}, v_{dq}$                                           & \checkmark                          &   \checkmark        &                   \checkmark     &      \\
dVOC \cite{colombino2017global, he2023nonlinear} & $i_{\ab}$ & $v_{\ab}$ & \cm & \cm & \cm &  \\
Matching (infinite bus) \cite[Sec. III]{arghir2019electronic}& $i_{dc}, i_{\ab}$         & $v_{dc}, v_{\ab}$                                           & \checkmark                          &   \checkmark        &                        &   \checkmark   \\
Matching (no infinite bus) \cite[Sec. IV]{arghir2018energy}& $i_{dc}, i_{\ab}$         & $v_{dc}, v_{\ab}$                                           & \checkmark                          &           &                 \checkmark       &   \checkmark   \\
Interlinking \cite{watson2021control} & $\omega, v_{dc}, v_{dq}$ & $i_{dc}, P, Q$ & & & \cm & \cm \\ 
       Bregman (droop laws) \cite{de2017bregman} 
       & $PQ$ setpoint & $(P, \frac{\partial \mathcal{S}}{\partial V})$ & \checkmark &           & \checkmark &  
\\ \hline\hline
\end{tabular}
\label{tab:comparison-passivity}
\end{table*}

This paper proves that the \ac{HAC} framework of \cite{tayyebi2023hac} satisfies large-signal passivity conditions under an appropriate choice of parameters. Thus, we show that HAC, in contrast to other available inverter control strategies, possesses network-independent large-signal passivity-guarantees The \ac{HAC} control strategy is a droop law that chooses a modulation based on the difference of the DC voltage from the setpoint (linearly) and the difference between the modulation phase angle and its corresponding setpoint (nonlinearly through a half-sine transformation). 
The stability proof is based on an incremental Lyapunov function comprising quadratic terms (in DC voltage, AC voltage, and AC inductor current) and a trigonometric term (in angle). By dropping the grid-side states of the  model in \cite{tayyebi2023hac}, we treat the inverter as an open-loop system. Thus, passivity can be proven using an incremental storage function adapted from the incremental Lyapunov function used in \cite{tayyebi2023hac} for the steady-state interconnected case. Retracing the stability arguments results in a proof of incremental passivity under very similar parametric conditions. This incremental passivity result implies stability of multi-inverter interactions between \ac{HAC}-controlled inverters over transmission lines composed of passive elements.
% \subsection{Paper Structure}
This paper has the following structure: Section~\ref{sec:preliminaries} reviews preliminaries such as  AC-DC power inverter models, HAC, and incremental passivity. Section~\ref{sec:hac_passive} derives large-signal and small-signal conditions under which \ac{HAC} is incrementally passive.
Section~\ref{sec:examples} presents the validation of the passivity result through an electromagnetic transient simulation of a multi-inverter power grid. Section \ref{sec:conclusion} concludes the paper.
\section{Preliminaries}
\label{sec:preliminaries}

\begin{acronym}

\acro{COI}{Center of Inertia}

\acro{dVOC}{Dispatchable Virtual Oscillator Control}

\acro{HAC}{Hybrid-Angle Control}

\acro{SMIB}{Single-Machine Infinite-Bus}

\acro{LMI}{Linear Matrix Inequality}
\acroplural{LMI}[LMIs]{Linear Matrix Inequalities}
\acroindefinite{LMI}{an}{a}

\acro{ODE}{Ordinary Differential Equation}
\acroindefinite{ODE}{an}{a}

\acro{PSD}{Positive Semidefinite}

\acro{PWM}{Pulse-Width Modulation}

\acro{SDE}{Stochastic Differential Equation}
\acroindefinite{SDE}{an}{a}

\acro{SDP}{Semidefinite Program}
\acroindefinite{SDP}{an}{a}

\acro{SOS}{Sum of Squares}
\acroindefinite{SOS}{an}{a}

\acro{VSC}{Voltage Source Converter}

\end{acronym}

\subsection{AC-DC Power Inverter Model and HAC}

The model considered for a three-phase interlinking AC/DC inverter in balanced operation was introduced in \cite{tabesh2008multivariable}. The H-bridge has a 2-level 6-switch topology. The representation of the inverter connected to the DC and AC grid is pictured in Figure \ref{fig:converter_model}.
The DC side is supplied with a current source $i_{\rm {dc}}$. The inverter has a conductance $G_{\rm {dc}}$ and a capacitance $C_{\rm {dc}}$ on its DC side. A current of $i_{\rm x}$ enters the H-bridge (switching), and a voltage of $v_{\rm x}$ is measured across the H-bridge on the AC side. A current of $i_{\alpha \beta}$ leaves the H-bridge, and flows through a resistor with resistance $R$ and an inductor with inductance $L$. A voltage of $v_{\alpha \beta}$ is measured across the AC-side capacitor $C$. An AC current of $i_{\ell}$ enters the AC grid load. Note that this averaged 2-level model can also be used as a base to represent other types of inverter topologies. As an example, \cite[Fig. 5]{gross2022dualport} presents a model for multi-modular inverter topologies.

\begin{figure}
    \centering
    \includegraphics{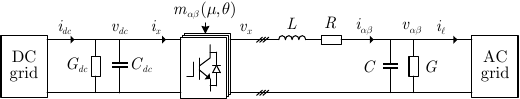}
    \caption{Averaged 2-level Converter Model}
    \label{fig:converter_model}
\end{figure}

The dynamics of the  inverter in the $\ab$ stationary reference frame are (Kirchoff's voltage and current laws) :
\begin{subequations}
\label{eq:ilc}
\begin{align}
    C_{\rm {dc}} \dot{v}_{\rm {dc}} &= -G_{\rm {dc}} v_{\rm {dc}} + i_{\rm {dc}} - i_{\rm {x}}, \\
    C \dot{v}_{\alpha \beta} &= -G v_{\alpha \beta} -i_{\ell} + i_{\alpha \beta}, \\
    L \dot{i}_{\alpha \beta} &= -R i_{\alpha \beta} - v_{\alpha \beta} + v_{\rm x}.
\end{align}
\end{subequations}
Letting $m_{\alpha \beta}(\cdot) \in [-1, 1]$ be a two-dimensional modulation signal and $\top$ denote the transpose of a matrix, the switching current and voltages $i_{\rm x}(t)$ and $v_{\rm x}(t)$ respectively can be modeled as 
\begin{align}
\label{eq:modulation_vi}
    i_{\rm x}(t) &= m_{\alpha \beta}(t)^\top i_{\alpha \beta}(t), & v_{\rm x}(t) &= m_{\alpha \beta} v_{\rm {dc}}(t).
\end{align}
The averaged two-level three-phase inverter model \eqref{eq:modulation_vi} possesses two bilinearities: between the applied modulation and the DC current, and between the applied modulation and the AC voltage.

The \ac{HAC} control law with grid frequency $\omega_0$, DC voltage setpoint $v_{\rm {dc}}^*$ , time-dependent angle setpoint $\theta^*(t)$, and modulation magnitude $\mu \in [0,1]$ determines the modulation $m_{\ab}(t)$ as \cite{tayyebi2023hac}
\begin{align}
    \dot{\theta}(t) &= \omega_0 + \eta(v_{\rm {dc}}(t) - v_{\rm {dc}}^*) - \gamma \sin \left(\frac{\theta(t) - \theta^*(t)}{2}\right) \label{eq:hac_law},\\
    m_{\alpha \beta}(t) &= \mu \begin{bmatrix}
        \cos(\theta(t))  & \sin(\theta(t))
    \end{bmatrix}^\top . \label{eq:modulation}
\end{align}

Specifically, the `hybrid-angle' portion of HAC is the half-angle-feedback law $ \gamma \sin((\theta - \theta^*)/2)$ in \eqref{eq:hac_law}.
\subsection{Incremental Passivity}

This section reviews the concepts of incremental passivity \cite{schiff1971passivity, van2000l2, pavlov2008incremental}. 
Consider a dynamical system with state $x$, input $u$, and output $y$:
\begin{align}
    \dot{x}(t) &= f(t, x(t), u(t)), \label{eq:system_test} & 
    y(t) &= h(x(t)).
\end{align}

\begin{defn}
    The system \eqref{eq:system_test} is \textbf{passive} if there exists a storage function $V_p(x)$ such that for any time horizon $T> 0$ and trajectory $x(t)$ it holds that
    \begin{align}
        V_p(x(T)) - V_p(x(0)) &< \int_{t=0}^T u(t)^\top y(t). \label{eq:passive_int}
        \intertext{If the storage function $V_p$ is continuously-differentiable, then the integral condition \eqref{eq:passive_int} can be replaced by the differential expression}
                \dot{V}_p(x(t)) &< u(t)^\top y(t) \label{eq:passive_int_diff}
    \end{align}
    for any input-output trajectory $(u(t), y(t))$.
\end{defn}

\begin{defn}
    The system in \eqref{eq:system_test} is \textbf{incrementally passive} if there exists a storage function $V(x, \bar{x})$ such that for any two trajectories $(x(t), \bar{x}(t))$ and time horizon $T>0$, the following relation holds:
    \begin{align}
        &V(x(T), \bar{x}(T)) - V(x(0), \bar{x}(0))\nonumber \\
        &\qquad <\int_{t=0}^T (u(t)-\bar{u}(t))^\top(y(t) - \bar{y}(t)) \ dt\label{eq:passive_int_incr}
    \end{align}
    for any two input-output trajectories $(u(t), y(t))$ and $(\bar{u}(t), \bar{y}(t))$.
\end{defn}

If the storage function $V(x, \bar{x})$ is once continuously differentiable, then the integral condition in \eqref{eq:passive_int_incr} can be equivalently treated as 
    \begin{align}
             \dot{V}(x(t), \bar{x}(t)) &< (u(t)-\bar{u}(t))^\top(y(t) - \bar{y}(t)). \label{eq:passive_int_diff_incr}
    \end{align}  

\begin{defn}
    An incrementally passive system is \textbf{output-strict incrementally passive} if there exists a positive-definite function $q(\cdot)$ such that 
    \begin{align}
             \dot{V}(x(t), \bar{x}(t)) &\leq (u(t)-\bar{u}(t))^\top(y(t) - \bar{y}(t))  - q(y(t) - \bar{y}(t))\label{eq:passive_int_strict}
    \end{align}    
    for any two input-output trajectories $(u(t), y(t))$ and $(\bar{u}(t), \bar{y}(t))$.
\end{defn}

\begin{defn}
\label{defn:error_coord}
    The \textbf{error coordinates} $\delta x = x(t) - \bar{x}(t)$, $\delta u = u(t) - \bar{u}(t)$, $\delta y = y(t) - \bar{y}(t)$,  associated with \eqref{eq:system_test} have the dynamics
\begin{align}
    \frac{d}{dt} \delta x(t) &= f(t, x(t), u(t)) - f(t, \bar{x}(t), \bar{u}(t)), \label{eq:system_test_delta} \\
    \delta y(t) &= h(x(t)) - h(\bar{x}(t)). \nonumber
\end{align}
\end{defn}

Passivity of \eqref{eq:system_test_delta} with input $\delta u,$ output $\delta y$, state $\delta x$, and storage $V_\delta(\delta x)$ function implies incremental passivity of \eqref{eq:system_test} with storage function $V(x, \bar{x}) = V_\delta(x - \bar{x}) = V_\delta(\delta x)$.

A linear system $\dot{x} = A_{\rm {lin}}x + B_{\rm {lin}}u, \ y = C_{\rm {lin}}x + D_{\rm {lin}}u$ is incrementally passive  if there exists a quadratic storage function $V(x) = \frac{1}{2}x^T P x$ such that \cite{kalman1962stability}
\begin{align}
    P &\succ 0,     & \begin{bmatrix}
        A_{\rm {lin}}^\top P + PA_{\rm {lin}}  & P B_{\rm {lin}} - C_{\rm {lin}}^\top \\ B_{\rm {lin}}^\top P - C_{\rm {lin}} & -D_{\rm {lin}}-D_{\rm {lin}}^\top
    \end{bmatrix} &\preceq 0.\label{eq:passivity_lmi}
\end{align}
A linear system is passive if and only if it is incrementally passive (by linearity). The Input Feedforward Passivity (IFP) and Output Feedback Passivity (OFP) indices of the linear system $G(s) = C (sI-A)^{-1}B + D$ are
\begin{subequations}
\begin{align}
    \text{IFP}(j \omega) &= \lambda_{\min}(G(e^{j \omega}) + G(e^{j \omega})^H)/2  \\
    \text{OFP}(j \omega) &= \lambda_{\min}(G(e^{j \omega})^{-1} + G(e^{j \omega})^{-H})/2.
\end{align}
\end{subequations}
The linear system is strictly input (output) passive if $\forall \omega \in [0, \infty): \ \text{IFP}(j \omega) > 0$ ($\text{OFP}(j \omega) > 0$).

\section{Incremental Passivity of HAC}
\label{sec:hac_passive}

This section  derives incremental passivity conditions for an inverter with dynamics in \eqref{eq:ilc} where the modulation $m_{\ab}$ is chosen according to the \ac{HAC}  law \eqref{eq:hac_law}. A DC-side current control with reference current $i^{ref}_{dc}$ and proportional DC-voltage gain $\kappa \geq 0$ ($i_{dc} = i_{\rm {dc}}^{\rm {ref}} + \kappa(v_{\rm {dc}}^* - v_{\rm {dc}}) $) is also considered, for which the case of $\kappa=0$ may be treated as a special case of an uncontrolled DC current. This closed-loop system has dynamics of 

\begin{subequations}
\label{eq:interlinked_closed}
    \begin{align}
    C_{\rm {dc}} \dot{v}_{\rm {dc}} &= -G_{\rm {dc}} v_{\rm {dc}} + (i_{\rm {dc}}^{\rm {ref}} + \kappa(v_{\rm {dc}}^* - v_{\rm {dc}})) \\
    & \qquad - \mu [\cos(\theta), \sin(\theta)]^\top i_{\alpha \beta}, \nonumber \\
    C \dot{v}_{\alpha \beta} &= -G v_{\alpha \beta} -i_{\ell} + i_{\alpha \beta}, \\
    L \dot{i}_{\alpha \beta} &= -R i_{\alpha \beta} - v_{\alpha \beta} + \mu  v_{\rm {dc}} [\cos(\theta), \sin(\theta)],  \\
    \dot{\theta} &= \omega_0 + \eta(v_{\rm {dc}} - v_{\rm {dc}}^*) - \gamma \sin \left(\frac{\theta - \theta^*(t)}{2}\right). \label{eq:angle_interlinked}
    \end{align}
\end{subequations}
For the moment, it is assumed that the modulation parameter $\mu$ is constant in time. The modulation  $\mu$ can become time-varying in $[0, 1]$ as in a signal $\mu(t)$ in order to enforce current limitation for safe inverter operation \cite{arghir2019electronic, tayyebi2023hac, he2024cross}. 
We  furthermore
define $\tilde{G}_{\rm {dc}}= G_{\rm {dc}} + \kappa$ as the effective conductance the DC element.

\subsection{Incremental Large-Signal Passivity of HAC}

Incremental passivity of \eqref{eq:interlinked_closed} is taken with respect to the following definitions:
\begin{subequations}
\label{eq:hac_set_states}
\begin{align}
    \text{States:} & & x &= [v_{\rm {dc}}, v_{\alpha \beta}, i_{\alpha \beta}, \theta] \\
    \text{Inputs:} & & u &= [i_{\rm {dc}}^{\rm {ref}}, -i_{\ell}] \\
    \text{Outputs:} & & y &= [v_{\rm {dc}}, v_{\alpha \beta}].
\end{align}
\end{subequations}

The main contribution of this paper is the following theorem of incremental passivity:
\begin{thm}
\label{thm:passive_hac}
    A sufficient condition for the  HAC-controlled inverter system in \eqref{eq:interlinked_closed} to be  incrementally passive with parameter gains $(\eta, \gamma)$ is if there exists $\epsilon_1, \epsilon_2 >0$ with
    \begin{align}
\label{eq:hac_passivity_conditions}        \epsilon_2^2 & < R,       \qquad \epsilon_1^2  < \tilde{G}_{\rm {dc}}/( \norm{\bar{i}_{\ab}} \mu)^2, \\
     (\lambda \eta/2)^2 & < \left(\lambda \gamma  - (1/\epsilon_1^2) - (\bar{v}_{\rm {dc}} \mu/\epsilon_2)^2\right) \left(\tilde{G}_{\rm {dc}}- (\epsilon_1 \norm{\bar{i}_{\ab}} \mu)^2\right). \nonumber
    \end{align}
    
\end{thm}

\begin{proof}
Let $x(t), \bar{x}(t)$ be any two trajectories of \eqref{eq:interlinked_closed}.
It is assumed that the reference values $\bar{v}_{\rm {dc}}$ and $\bar{\theta}(t)$ are known by the inverter (as the setpoints $v_{\rm {dc}}^*, \theta^*(t)$) and satisfy the power flow equations, but that all other terms in $\bar{x}(t)$ are unknown.
Furthermore, let $\psi$ denote the trigonometric term $\psi(\theta) = [\cos(\theta), \sin(\theta)]^\top$ (in which the $\theta$ dependence will be omitted).
The error dynamics (via Definition \ref{defn:error_coord}) of $\delta{x} = x - \bar{x}(t)$ with $\delta{\psi} = \psi(\theta) - \psi(\bar{\theta})$ are given as
\begin{subequations}
\label{eq:interlinked_closed_error}
    \begin{align}
    C_{\rm {dc}}\frac{d\delta{v}_{\rm {dc}}}{dt} &= -\tilde{G}_{\rm {dc}} \delta{v}_{\rm {dc}} + \delta{i}_{\rm {dc}}^{\rm {ref}} - \mu  (\delta \psi^\top \bar{i}_{\alpha \beta} + \psi^\top \delta{i}_{\alpha \beta}), \\
    C \frac{d\delta{v}_{\ab}}{dt} &= -G \delta{v}_{\ab} + \delta{i}_{\ab} - \delta{i}_{\ell}, \\
    L \frac{d\delta{i}_{\ab}}{dt}&= -R \delta{i}_{\ab} - \delta{v}_{\ab} + \mu (\psi \delta{v}_{\rm {dc}} + \delta{\psi}\bar{v}_{\rm {dc}}), \\
    \frac{d\delta{\theta}}{dt} &= \eta \delta{v}_{\rm {dc}} - \gamma \sin(\delta{\theta}/2).
    \end{align}
\end{subequations}

An incremental storage function candidate for dynamics in \eqref{eq:interlinked_closed_error} with a free parameter $\lambda \in \R_{>0}$ is
\begin{align}
    V(\delta{x}) =& \frac{1}{2}(C_{\rm {dc}} \delta{v}_{\rm {dc}}^2 +C \norm{\delta{v}_{\alpha \beta}}^2 + L \norm{\delta{i}_{\alpha \beta}}^2)\nonumber\\
    &+ 2 \lambda(1-\cos(\delta{\theta}/2)).\label{eq:incremental_storage}
\end{align}

The time derivative of the incremental storage function $V(\delta x)$ from \eqref{eq:incremental_storage} is 

\begin{align}
    \dot{V}(\delta {x}) &= -\tilde{G}_{\rm {dc}} \delta{v}_{\rm {dc}}^2 - G \norm{\delta{v}_{\ab}}_2^2 - R \norm{\delta{i}_{\ab}}_2^2 \label{eq:hac_storage_vdot}\\
    &+ \delta{v_{\rm {dc}}} \delta i_{\rm {dc}}^{\rm {ref}} + \delta{v_{\ab}}^\top (-\delta{i_{\ell}}) \nonumber\\
    &+ \mu(\delta{i}_{\ab}^\top \delta{\psi} \, \bar{v}_{\rm {dc}} - (\delta{\psi}^\top \, \bar{i}_{\ab})  \delta{v}_{\rm {dc}}) \nonumber \\ 
    &+ \lambda \eta  \ \delta{v}_{\rm {dc}} \sin(\delta{\theta}/2) - \lambda \gamma \sin^2(\delta{\theta}/2).  \nonumber
\end{align}
Young's inequality can be used to upper-bound the storage function \eqref{eq:hac_storage_vdot} 
while only knowing bounds on $\bar{v}_{\rm {dc}}$ and $\norm{\bar{i}_{\ab}}$.

Thus, the problematic terms $(\bar{v}_{\rm {dc}}, \bar{i}_{\alpha \beta})$ can be replaced with their norms and with values $\epsilon_1, \epsilon_2$ as (from equations (17a) and (17b) of \cite{tayyebi2023hac} with $\norm{\delta \psi}_2^2 = 4\sin^2(\delta \theta/2).$):
    \begin{align}
     (-\delta{\psi})^\top \, (\mu \bar{i}_{\ab}  \delta{v}_{\rm {dc}}) &\leq ( \mu \norm{\bar{i}_{\ab}} \epsilon_1)^2 \delta v_{\rm {dc}}^2 + 1/(4 \epsilon_1^2) \norm{\delta \psi}_2^2, \nonumber\\
     &=  ( \mu \norm{\bar{i}_{\ab}} \epsilon_1)^2 \delta v_{\rm {dc}}^2 + 1/(\epsilon_1^2) \sin^2(\delta \theta/2), \nonumber\\
        (\delta{i}_{\ab})^\top \, (\mu \delta{\psi} \, \bar{v}_{\rm {dc}}) &\leq (\epsilon_2 )^2 \norm{\delta i_{\ab}}^2_2 + (\mu \bar{v}_{\rm {dc}}/\epsilon_2)^2/4  \norm{\delta \psi}_2^2 \nonumber \\
        &=(\epsilon_2 )^2 \norm{\delta i_{\ab}}^2_2 + (\mu \bar{v}_{\rm {dc}}/\epsilon_2)^2  \sin(\delta \theta/2)^2.  
    \end{align}
Under the definition of 
\begin{align}
    \zeta = [\delta v_{\rm {dc}}, \delta v_{\ab}, \delta i_{\ab}, \sin(\delta \theta/2)], \label{eq:zeta}
\end{align} the upper-bound on $\dot{V}$ can be expressed as 
\begin{subequations}
\label{eq:vdot_upper_expand_dc}
    \begin{align}
        \dot{V}(\delta{x}) &\leq  -\zeta^\top \mathcal{Q} \zeta + y^\top u, \\
        \mathcal{Q} &= \begin{bmatrix}
            \tilde{G}_{\rm {dc}}  - (\epsilon_1 \norm{\bar{i}_{\ab}}\mu)^2 & 0 & 0 & -\lambda \eta/2 \\
            0 & G I & 0 & 0 \\
            0 & 0 & (R - \epsilon_2^2)I & 0 \\
            -\lambda \eta/2 & 0 & 0 & \Lambda
        \end{bmatrix} \label{eq:vdot_upper_Q_dc} \\
        \Lambda &= \lambda \gamma  - (1/\epsilon_1^2) - (\mu\bar{v}_{\rm {dc}} /\epsilon_2)^2. \label{eq:vdot_upper_lam}        
    \end{align}
\end{subequations}

The conditions in \eqref{eq:hac_passivity_conditions} are an equivalent expression for the constraint $\mathcal{Q} \succ 0$ from \eqref{eq:vdot_upper_Q_dc}.

\end{proof}

\begin{rmk}
    The conditions in \eqref{eq:hac_passivity_conditions} can be interpreted as follows: The first two inequalities ensure that the dissipation in the AC filter  (resp. DC link) is sufficiently large when compared with the worst case power exchange on the AC  (resp. DC) side of the inverter. The third inequality insists that the net damping injected into the angle dynamics must be larger than the residual non-dissipative couplings from the DC and AC subsystems.
\end{rmk}

\begin{rmk}
    The incremental passivity result in Theorem \ref{thm:passive_hac} is sufficient but not necessary to prove the incremental stability of a system under passive interconnection. Condition \eqref{eq:hac_passivity_conditions} is a decentralized statement valid on each HAC-inverter, because the stability certificate is independent of the precise interconnecting topology and admittance.    
\end{rmk}

\subsection{Incremental Small-Signal Passivity of HAC}
% \label{sec:passivity_small}

This subsection  explores the small-signal passivity properties of \ac{HAC}.  The system is linearized in the global $dq$ coordinates about a rotating frame with constant frequency $\omega_0$. The angle is therefore written in a rotating frame as $\theta \leftarrow \theta - \omega_0 t$. We define $J$ as the rotation matrix:
\begin{align}
    J = \begin{bmatrix}
        0 & 1 \\ -1 & 0
    \end{bmatrix}.
\end{align}The error dynamics \eqref{eq:interlinked_closed_error} in the $dq$ coordinates are \begin{subequations}
\label{eq:interlinked_closed_dq}
    \begin{align}
    C_{\rm {dc}}\frac{d\delta{v}_{\rm {dc}}}{dt} &= -\tilde{G}_{\rm {dc}} \delta{v}_{\rm {dc}} + \delta{i}_{\rm {dc}}^{\rm {ref}} - \mu  (\delta \psi^\top \bar{i}_{\rm {dq}} + \psi^\top \delta{i}_{\rm {dq}}), \\
    C \frac{d\delta{v}_{\rm {dq}}}{dt} &= -G \delta{v}_{\rm {dq}} + \delta{i}_{\rm {dq}} - \delta{i}_{\ell} - C \omega_0 J \delta v_{\rm {dq}},\\
    L \frac{d\delta{i}_{\rm {dq}}}{dt}&= -R \delta{i}_{\rm {dq}} - \delta{v}_{\rm {dq}} + \mu (\psi \delta{v}_{\rm {dc}} + \delta{\psi}\bar{v}_{\rm {dc}})  \nonumber\\
    & \qquad - L \omega_0 J \delta i_{\rm {dq}},\\
    \frac{d\delta{\theta}}{dt} &= \eta \delta{v}_{\rm {dc}} - \gamma \sin(\delta{\theta}/2).
    \end{align}
\end{subequations}

% Linearization about the setpoint $(v_{\rm {dc}}^*, v_{\rm {dq}}^*, i_{\rm {dq}}^*, \theta^*)$ results in the dynamics of \urg{Xiuqiang: help?}
% \begin{figure}[h]
%     \centering
%     \includegraphics[width=\linewidth]{img/image_dq_derive.jpeg}
%     \caption{\urg{Xiuqiang's notes on the derivation}}
%     \label{fig:enter-label}
% \end{figure}
The error dynamics in \eqref{eq:interlinked_closed_dq} can be linearized about the equilibrium $(v_{\rm {dc}}^*, v_{\rm {dq}}^*, i_{\rm {dq}}^*, \theta^*)$, in which the setpoint can be interpreted as the steady-state alternate trajectory $\bar{x}(t)$ in the $dq$ frame. Defining the half-angle $\phi$ with  $\delta \phi = \delta \theta /2$, the  linearization of \eqref{eq:interlinked_closed_dq}  can be described using the incremental states, inputs,  and outputs of
% The linearization of \eqref{eq:interlinked_closed_dq} can be described using in the shifted input-state-outputs of 
\begin{subequations}
\begin{align}
    x_{\rm {lin}} &= \begin{bmatrix}
        \delta v_{\rm {dc}} & \delta v_{\rm {dq}} & \delta i_{\rm {dq}}^{\rm {ref}} & \delta \phi
    \end{bmatrix}^\top, \\
    u_{\rm {lin}} &= \begin{bmatrix}
        \delta i_{\rm {dc}} &  -\delta i_\ell
    \end{bmatrix}^\top, \\
    y_{\rm {lin}} &= \begin{bmatrix}
        \delta v_{\rm {dc}} & \delta v_{\rm {dq}}
    \end{bmatrix}^\top.
\end{align}
\end{subequations}
We define the steady-state trigonometric lift of the angle setpoint $\theta^*$ as  $\psi_* = [\cos(\theta^*)\,  \sin(\theta^*)]^\top$. Furthermore we also define a dynamics matrix $\bar{A}$ and the scaling matrix $S$ as
\begin{subequations}
\begin{align}
   \bar{A} &=  \\
   & \begin{bmatrix}
        -\tilde{G}_{\rm {dc}} & 0 & -\mu \psi^\top_* & -2 \mu  (J \psi_*)^\top i_{\rm {dq}}^* \\
        0 & -G I - C \omega_0 J & I & 0 \\
        \mu \psi_* & -I & -R I - L \omega_0 J & 2 \mu v_{\rm {dc}}^* J \psi_* \\
        \eta/2 & 0 & 0 & -\gamma/2
    \end{bmatrix}, \nonumber \\
    S &=\textrm{diag}(C_{\rm {dc}}, C, C, L, L, 1).
\end{align}
\end{subequations}
The linearized dynamics of \eqref{eq:interlinked_closed_dq} can be therefore written as
\begin{align}
    \dot{x}_{\rm {lin}} &= A_{\rm {lin}} x_{\rm {lin}} + B_{\rm {lin}} u_{\rm {lin}} \label{eq:interlinked_closed_dq_lin}\\
    y_{\rm {lin}} &= C_{\rm {lin}}x_{\rm {lin}} + 0 u_{\rm {lin}},\nonumber
\end{align}
under the definitions
% \begin{subequations}
\begin{align}
    A_{\rm {lin}} &= \label{eq:linearized_matrices}
    S^{-1} \bar{A},  \\
    B_{\rm {lin}} &= \begin{bmatrix}
        1/C_{\rm {dc}} & 0 \\
        0 & (1/C) I \\
        0 & 0 \\
        0 & 0
    \end{bmatrix}, \quad 
    C_{\rm {lin}} = \begin{bmatrix}
        1 & 0 & 0 & 0 \\
        0 & I & 0 & 0
    \end{bmatrix}.    \nonumber 
\end{align}
% \end{subequations}

We now proceed to provide conditions for passivity of \eqref{eq:interlinked_closed_dq_lin}, \eqref{eq:linearized_matrices} via the LMI expression in \eqref{eq:passivity_lmi}.
A quadratic storage function $V(x_{\rm {lin}})$ can be defined as 
\begin{align}
    V(x_{\rm {lin}}) &= \frac{1}{2}x_{\rm {lin}}^\top P x_{\rm {lin}}, 
 \label{eq:passivity_storage}\\
    P &= \begin{bmatrix}
        C_{\rm {dc}}  & 0 & 0 & 0 \\
        0 & C I & 0 & 0 \\
        0 & 0 & L I & 0 \\
        0 & 0 & 0 & 2\lambda
    \end{bmatrix}. \label{eq:passivity_storage_P}
\end{align}

Noting that the matrices in \eqref{eq:linearized_matrices} and \eqref{eq:passivity_storage_P} satisfy
\begin{align}
    B_{\rm {lin}}^\top P - C_{\rm {lin}} &= 0, & D_{\rm {lin}} = 0,
\end{align}
and that $P \succ 0$ if $\lambda > 0$, the passivity constraint via the expression in \eqref{eq:passivity_lmi} can be placed in the form of
    $(A_{\rm {lin}}^\top P + P A_{\rm {lin}})/2 \preceq 0$.
    
    After defining the quantity $\tau_*(\lambda, \eta) = (\lambda \eta - 2\mu(J \psi_*)^\top i_{\rm {dq}}^*)/2$, the above  LMI can be written as
    \begin{align}    
    \begin{bmatrix}
        \tilde{G}_{\rm {dc}} & 0 & 0 &  \tau_*(\lambda, \eta)\\
        0 & GI & 0 & 0 \\
        0 & 0 & RI & -\mu v_{\rm {dc}}^* J \psi_* \\
        \tau_*(\lambda, \eta) & 0 & -\mu v_{\rm {dc}}^* (J \psi_*)^\top  &         \lambda \gamma
    \end{bmatrix} \succeq 0. \label{eq:lmi_con}
\end{align}

% The procedure used in \eqref{eq:eps_bound} can be employed to create passivity conditions based on $\norm{i^*_{\rm {dq}}}$ in the case where $i^*_{\rm {dq}}$ is unknown (noting that $\norm{J \psi_*} = 1$):
In the case where $\psi_*$ and $i^*_{\text{dq}}$ are unknown while the magnitude $\norm{i^*_{\rm {dq}}}$  is known  (noting that $\norm{J \psi_*} = 1$), the off-diagonal terms involving $\psi_*$ and $i^*_{\text{dq}}$ can be bounded by Young's inequality:
   \begin{align}
    \mu(J \psi_*^\top \, i^*_{\rm {dq}})  \delta{v}_{\rm {dc}} (2\delta \theta) &\leq ( \mu \norm{i_{\rm {dq}}^*} \epsilon_1)^2 \delta v_{\rm {dc}}^2 + (2 \delta \theta)^2/(4 \epsilon_1^2),   \nonumber \\
        \mu(\delta{i}_{\rm {dq}}^\top \, J \psi_*) (2 \delta \theta) \, \bar{v}_{\rm {dc}} &\leq (\epsilon_2 )^2 \norm{\delta i_{\rm {dq}}}^2_2 + (2 \delta \theta)^2(\mu \bar{v}_{\rm {dc}}/\epsilon_2)^2/4. \label{eq:eps_bound_lin}
    \end{align}

A lower-bound  on $\dot{V}_{\rm {lin}}$ can therefore be found as 
\begin{subequations}
\label{eq:vdot_upper_expand_lin}
    \begin{align}
        \dot{V}_{\rm {lin}}(\delta{x}) &\leq  -x_{\rm {lin}}^\top \mathcal{Q}_{\rm {lin}} x_{\rm {lin}} + y_{\rm {lin}}^\top u_{\rm {lin}}, \\    
        \mathcal{Q}_{\rm {lin}} &= \begin{bmatrix}
            \tilde{G}_{\rm {dc}} - (\epsilon_1 \norm{i^*_{\rm {dq}}}\mu)^2 & 0 & 0 & -\lambda \eta/2 \\
            0 & GI & 0 & 0 \\
            0 & 0 & (R - \epsilon_2^2)I & 0 \\
           - \lambda \eta/2 & 0 & 0 & \Lambda_{\rm {lin}}
        \end{bmatrix} \label{eq:vdot_upper_Q_lin} \\
        \Lambda_{\rm {lin}} &= \lambda \gamma  - (1/\epsilon_1^2) - (\bar{v}_{\rm {dc}} \mu/\epsilon_2)^2. \label{eq:vdot_upper_lam_lin}
    \end{align}
\end{subequations}

Passivity conditions for the linearized system in $x_{\text{lin}}$ are
% \begin{subequations}
    \begin{align}
\label{eq:hac_passivity_conditions_lin}        &\epsilon_2^2  < R,       \qquad \epsilon_1^2  < \tilde{G}_{\rm {dc}}/( \norm{\bar{i}_{\rm {dq}}} \mu)^2, \\
 &   (\lambda \eta/2)^2  <   \left(\lambda \gamma  - (1/\epsilon_1^2) - (\bar{v}_{\rm {dc}} \mu/\epsilon_2)^2\right) \left(\tilde{G}_{\rm {dc}} - (\epsilon_1 \norm{\bar{i}_{\rm{dq}}} \mu)^2\right), \nonumber
    \end{align}
    which are identical to the large-signal conditions in \eqref{eq:hac_passivity_conditions} for the nonlinear system in $x$.
\section{Simulation Validation}

\label{sec:examples}

\begin{table*}
\centering
\caption{Parameters in the Simulation Study}
 \begin{minipage}{\linewidth}
 \begin{center}
\begin{tabular}[t]{lllll}
\arrayrulecolor{black}
\hline \hline
Symbol                 & Description       & Inv. 1 & Inv. 2 & Inv. 3               \\
\hline
$S_N$       & Nominal capacity  & $247.5$\,MVA\,& $192$\,MVA & $128$\,MVA \\
$V_{ll}$    & Nominal Line-to-Line Voltage & $0.69$\,kV & $0.69$\,kV & $0.69$\,kV\\
$V_{dc}$    & Nominal DC Voltage & $1.13$\,kV & $1.13$\,kV & $1.13$\,kV\\
$C_{dc}$ & DC Capacitance & $8.07$,\,F & $14.44$\,F & $5.78$\,F \\
$G_{dc}$ & DC Conductance & $0.19$,\,S & $0.15$\,S & $0.10$\,S \\
$L_\mathrm{f}$      & Filter inductance & $0.05$\,p.u. & $0.05$\,p.u. & $0.05$\,p.u. \\
$R_\mathrm{f}$      & Filter resistance & $0.05/30$\,p.u. & $0.05/30$\,p.u. & $0.05/30$\,p.u.\\
$C_\mathrm{f}$      & Filter capacitance & $0.05$\,p.u. & $0.05$\,p.u. & $0.05$\,p.u. \\
$\theta^{\star}$ & Angle setpoint & $0.0108$\,rad & $0.0108$\,rad & $0.0108$\,rad\\
$v_{dc}^{\star}$ & DC voltage setpoint & $1.00$\,p.u.& $1.00$\,p.u. & $1.00$\,p.u. \\
$\eta$ & DC voltage Gain & $0.001$\,rad/sV & $0.001$\,rad/sV &$0.001$\,rad/sV \\
$\gamma$ & Angle term Gain & $100$\,rad/s & $100$\,rad/s & $100$\,rad/s \\
$\kappa$ & DC-side gain & $1.9494 \times 10^{4}$\,S & $1.5123 \times 10^{4}$\,S & $1.0082 \times 10^{4}$\,S \\
\hline \hline
\end{tabular}
 \end{center}
\end{minipage}
\label{tab:system-parameters}
\end{table*}
This section validates the passivity properties of the \ac{HAC} with the help of a electromagnetic transient simulation \footnote{\url{https://doi.org/10.3929/ethz-b-000705638}}. We consider the IEEE 9-bus test system shown in Figure~\ref{fig:simulation_system}. The synchronous machines in the original system are replaced by \ac{HAC} inverters.
\begin{figure}[h]
    \centering
    \includegraphics{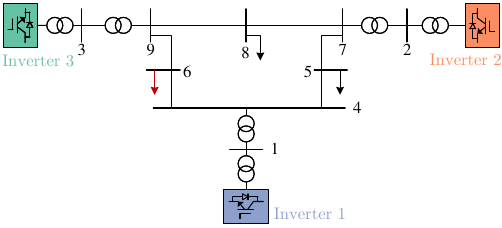}
    \caption{IEEE 9-bus system with HAC-controlled inverters connected to buses 1, 2 and 3.}
    \label{fig:simulation_system}
\end{figure}

The parameters of the system are provided in Table~\ref{tab:system-parameters} in SI units or the per-unit system (calculated on the respective base values). It should be noted that the gains of the \ac{HAC} are tuned to obey the nonlinear decentralized passivity condition in \eqref{eq:hac_passivity_conditions}. The parameters and certificates for inverter 3 are
\begin{align*}
    \eta &= 10^{-3} & \gamma &= 10^2 & \kappa &= 1.0082 \times 10^{4} \\
    \lambda &= 10^{10} & \epsilon_1 &= 2.2097\times 10^{-4} & \epsilon_2 &= 1.4375\times 10^{-3}.
\end{align*}
Figure \ref{fig:passivity_small_signal} plots the non-negative IFP and OFP for inverter 2 as a function of frequency. Output-strict passivity is achieved, because the passivity indices are always bounded above 0.
\begin{figure}[h]
    \centering
    \includegraphics[width=0.85\linewidth]{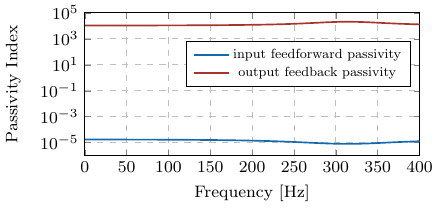}
    \caption{Small-Signal Passivity of Inverter 2}
    \label{fig:passivity_small_signal}
\end{figure}

To test the system, the load at bus~6 is \emph{doubled} at $t=1.5$ seconds from an initial consumption of 125 MW and 50 MVAr. As seen from Figure~\ref{fig:Area1_Results}, the system stabilizes following the load disturbance owing to the fact that the gains are tuned to ensure incremental passivity of the \ac{HAC} law.  
\begin{figure}
    \centering
    \includegraphics[width=\linewidth]{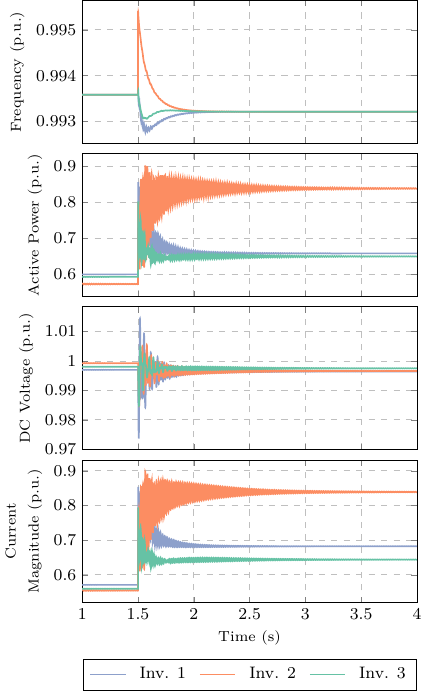}
    \caption{Simulation results for a load disturbance at Bus 6.}
    \label{fig:Area1_Results}
\end{figure}
\section{Conclusion}

\label{sec:conclusion}

This work proved that the grid-forming \ac{HAC} scheme of \cite{tayyebi2023hac} can satisfy an incremental passivity property from the current to the voltage. It explored the large-signal and small-signal passivity properties of \ac{HAC} and presented decentralized parametric conditions for which passivity is ensured. The theoretical findings of passivity were  confirmed in simulation.

The dominant obstacle of HAC as a grid-forming control method is the requirement that the angle setpoint $\theta^*(t)$ must be known. Future work will investigate how to approximate the angle setpoint from recorded $i/v$ signals, for instance, from a Phasor Measurement Unit (PMU), while retaining passivity guarantees. Other opportunities include the investigation of passivity for single-phase or unbalanced networks, determining decentralized passivity guarantees under current limits (through reachability methods), and experimentation on real physical systems.

\section*{Acknowledgements}

The authors would like to thank Florian D\"{o}rfler, Linbin Huang, Ali Tayebbi, Adolfo Anta, Irina Subot\'{i}c, Dominic Gross, Saverio Bolognani, Jean-Sebastian Broullion, and Catalin Arghir for discussions about grid-forming control.

\bibliographystyle{IEEEtran}
\bibliography{references}

\end{document}